\documentclass[journal]{IEEEtran}
\usepackage{amsthm,amsmath,amssymb}
\usepackage{graphicx}
\usepackage{url}
\usepackage{algorithm}
\usepackage{algorithmic}
\usepackage{color}
\usepackage{verbatim}
\usepackage{float}

\ifCLASSINFOpdf
\else
\fi
\begin{document}
%
\title{Multi-User Cooperative Computation Framework Based on Bertrand Game}
%
%
%
\author{Nan~Zhang, Guopeng~Zhang, Kezhi~Wang and Kun~Yang
\thanks{Nan Zhang and Guopeng Zhang are with the School of Computer
Science and Technology, China University of Mining and Technology, Xuzhou
221116, China (e-mail: NanZhang@cumt.edu.cn; gpzhang@cumt.edu.cn).}
\thanks{Kezhi Wang is with the Department of Computer and Information
Science, Northumbria University, Newcastle NE2 1XE, U.K. (e-mail:
kezhi.wang@northumbria.ac.uk).}
\thanks{Kun Yang is with School of Information and Communication Engineering, University of Electronic Science and Technology of China, Chengdu, China, and also affiliated with School of Computer Technology and Engineering, Changchun Institute of Technology, Changchun, China.
kunyang@uestc.edu.cn}

\thanks{ACK: Natural Science Foundation of China (Grant No. 61620106011, U1705263), UESTC Yangtze Delta Region Research Grant (Grant No.: 2020D002).} }

%
%


\maketitle

\begin{abstract}
In this paper, a multi-user cooperative computing framework is applied to enable mobile users to utilize available computing resources from other neighboring users via direct communication links. An incentive scheme based on Bertrand game is proposed for the user to determine \textit{who} and \textit{how} to cooperate. We model the resource demand users as \textit{buyers} who aim to use minimal payments to maximize energy savings, whereas resource supply users as \textit{sellers} who aim to earn payments for their computing resource provision. A Bertrand game against \textit{buyer's market} is formulated. When the users have \textit{complete information} of their opponents, the Nash equilibrium (NE) of the game is obtained in closed form, while in the case of \textit{incomplete information}, a distributed iterative algorithm is proposed to find the NE. The simulation results verify the effectiveness of the proposed scheme.
\end{abstract}

\begin{IEEEkeywords}
Cooperative computation framework, Task offloading, Resource allocation, Bertrand game, Nash equilibrium.
\end{IEEEkeywords}

\IEEEpeerreviewmaketitle

\section{Introduction}
%
%
%
%
\IEEEPARstart{W}{ith} the development of Artificial Intelligence (AI) applications, Virtual Reality (VR) services and others, the computational complexity of mobile applications is increasing
rapidly. Although cloud computing can alleviate the shortage
of computing resources for user equipments (UEs), it may
incur longer transmission delay. To reduce latency and increase user experience, the framework of multi-user cooperative computing has been proposed to enable UEs to offload their computational-intensive tasks to other neighbouring UEs via direct communication links \cite{7996590}.

In the cooperative computing framework, UEs are generally divided into resource demand UEs (DUs) and resource supply UEs (SUs). DUs have heavy computing tasks and may require other devices to assist, while SUs have available resources and may help others. In \cite{8644186}, the authors considered that a DU can partition a computing task into multiple parts and offload the task parts to a set of SUs. The purpose was to maximize the completion rate. The authors in \cite{8779699} proposed an energy-efficient task offloading method for two cooperative UEs and the aim was to minimize the long-term energy consumption of both users. In \cite{7996590}, the authors proposed a hybrid task offloading framework for \textit{fog computing}, where UEs can choose flexibly from local computing, another device, or the central cloud. In \cite{8360883}, the authors proposed that a SU can make a decision on whether to provide computing service according to the social relationship with a DU.

The provision of computing services to others incurs resource consumption for \textit{rational} providers, who may be only willing to use resources for profit purposes. However, the above works either consider UEs to be fully cooperative \cite{7996590}-\cite{8779699}\cite{7161307} or the cooperation depends on their social relations \cite{8360883}. In this paper, we propose an incentive scheme based on Bertrand game \cite{10.2307/2555525} to address the above-mentioned issue. Bertrand game is widely used to analyze price competition behavior and product sales share of \textit{buyer’s market}. 
Compared to existing works, the main contribution of this paper includes:

\begin{enumerate}
    \item By modeling DUs as \textit{resource buyers} and SUs as \textit{resource sellers}, a Bertrand game is proposed to solve the incentive problem in cooperative computing framework.
    \item When the states and actions of all users are observable (known as the \textit{complete information} hypothesis), the NE of the game is solved in closed-form. In the case of \textit{incomplete information}, an effective distributed iterative algorithm is proposed to search the NE of the game.  
\end{enumerate}

\section{System Model}
Consider that there are a set $\mathcal{M}=\left\{1,\cdots,M\right\}$ of $M$ UEs and each pair of the UEs can communicate with each other directly. A discrete time model is considered and let $t=1,2,\cdots$ denote the consecutive time slots. The duration of a time slot is denoted by $T$. In any slot $t$, UE $m$ $\left(\forall m\in\mathcal{M}\right)$ has a computing task to be executed, which is characterized by a tuple of two parameters $\varphi_m=\left(L_m,C_m\right)$, where $L_m$ (in Mb) is the amount of the task input data and $C_m$ (in CPU cycles/Mb) is the required number of CPU cycles to complete each megabit data \cite{8779699}. Let $f_m^\text{max}$ denote the maximum operating frequency of the CPU of UE $m$. In order to complete task $\varphi_m$ within a time slot, the CPU frequency is set to $f_m=\frac{C_m L_m}{T}$ ($0< f_m\leq f_m^\text{max}$) by UE $m$. Let $\kappa_m$ denote the effective capacitance coefficient of the CPU of UE $m$. The energy spent by UE $m$ to complete task $\varphi_m$ is given by \cite{8779699}

\begin{small}
\begin{equation}
    E_{m}^\text{exe}=\kappa_m f_{m}^2 C_m L_m=\frac{\kappa_m \left(C_m L_m\right)^3}{T^2 },\ \forall m\in\mathcal{M}.
    \label{beforeenergy}
\end{equation}
\end{small}


We consider the situation of \textit{buyer's market}. Let $m=0$ denote the unique DU in the system. The other UEs in set $\mathcal{M}-\left\{0\right\}$ are the potential SUs for the DU. We also define a new set $\mathcal{N}=\left\{1,\cdots,N\right\}\subseteq\mathcal{M}-\left\{0\right\}$ to represent the set of SUs that actually provide computing service to the DU. 

\subsection{Energy consumption of the DU}
The energy consumption of the DU consists of the following two parts. The first part of energy is used to transmit the task data to the SUs in set $\mathcal{N}$ in slot $t-1$, and the second part of energy is used to execute the remaining task locally in slot $t$. Note that the above system settings allows task data transmission and task computing to be carried out in parallel.

In slot $t-1$, we assume that any SU $i$ in set $\mathcal{N}$ is allocated equal time for receiving the task data, which is given by
\begin{equation}
    t_n=T/\left|\mathcal{N}\right|,\ \forall n \in \mathcal{N}. 
\end{equation}

Let $p_0^n$ denote the transmit power of the DU for uploading task data to SU $n$. Let $g_0^n$ denote the channel gain from the DU to SU $n$. Then, the achievable data rate is given by
\begin{equation}
    r_0^n=B\log_2\left(1+\frac{p_0^n g_0^n}{\sigma^2}\right),\ \forall n \in \mathcal{N},
    \label{datarate}
\end{equation}
where $B$ is the channel bandwidth of the system, and $\sigma^2$ is the noise power at the receiver of the SU.

The unique DU in the system is denoted by $m=0$. Then $L_0$ denotes the amount of task input data of the DU. Let $l_0^n$ $\left(0\leq l_0^n\leq L_0\right)$ denote the amount of task data that the DU decides to offload to SU $n$. To ensure correct reception, condition $r_0^n t_n\geq l_0^n$ should be satisfied. By substituting eq. \eqref{datarate} into this condition, the lower bound of $p_0^n$ is obtained as
\begin{equation}
    p_0^n=\left(2^{\frac{l_0^n}{BT/\left|\mathcal{N}\right|} }-1\right) \frac{\sigma^2}{g_0^n},\ \forall n\in{\mathcal{N}}. 
\end{equation}
Then, the energy spent by the DU for task offloading in slot $t-1$ is given by

\begin{equation}
    E_0^{\text{off}}=\sum\nolimits_{n\in \mathcal{N}}p_0^n t_n.
\end{equation}

Assume that the CPU of the DU always works at the highest frequency $f_0^\text{max}$, although the task may be offloaded, it still needs to complete other new tasks. Hence, the energy spent by the DU for executing the remaining task in slot $t$ is
\begin{equation}
    E_0^\text{com}=\kappa_0 (f_0^\text{max} )^2 C_0 \left(L_0-\sum\nolimits_{n\in \mathcal{N}}l_0^n\right).
\end{equation}

\subsection{Energy consumption of each SU}
The energy consumed by any SU $n$ in set $\mathcal{N}$ consists of the following two parts. The first part is used by SU $n$ for receiving the task data of the DU in slot $t-1$. Let $p_n^\text{rec}$ denote the power of the receiver circuit of SU $n$. This part of energy consumption is given by
\begin{equation}
    E_n^\text{rec}=p_n^\text{rec} t_n, \ \forall n \in \mathcal{N}.
\end{equation}

After decoding the task data of the DU, SU $n$ has to raise the CPU frequency in slot $t$ to complete its own task $\varphi_n$ and the DU's task of $l_0^n$ Mb within slot $t$. By using eq. \eqref{beforeenergy}, we can derive the energy consumption of SU $n$ in slot $t$ as
\begin{equation}
    E_n^\text{com}=\frac{\kappa_n C_n^3\left(L_n+l_0^n \right)^3}{T^2 }, \ \forall n \in \mathcal{N}. 
\end{equation}

Since the size of the execution result is small, the energy spent by SUs for feeding back the result is negligible.

\section{Game Formulation}
Our objective is to motivate \textit{rational} UEs to participate in cooperative computing framework. Therefore, the following questions should be answered \cite{4684603}: 1) \textit{who to cooperate}, i.e., to determine which SUs are in set $\mathcal{N}$, and 2) \textit{how to cooperate}, i.e., to determine how much resource a SU in set $\mathcal{N}$ provides to the DU and how much it can benefit from the provision. Next, we propose a Bertrand game based incentive scheme to address these issues. 

In the proposed game, the DU is modeled as a \textit{buyer} who aims to use minimal payment to maximize energy savings. Generally, the benefits of a player in a game is quantified by utilities. Let $q_n$ represent the energy pricing of SU $n$ for providing computing resources to perform 1 Mb computing tasks of the DU. According to \cite{10.2307/2555525}, the following function is used to quantify the utility of the DU.

\begin{align}
    U_0=\left(E_0^\text{exe}-E_0^\text{com}-E_0^\text{off}\right)-\sum\nolimits_{n\in \mathcal{N}}q_nl_0^n- \label{U0} \\ \notag \frac{1}{2}\left(\sum\nolimits_{n\in \mathcal{N}}\left(l_0^n\right)^2+2v\sum\nolimits_{n\neq k}l_0^nl_0^k\right),
\end{align}
where the \textit{first} term represents the energy saving for the DU when it performs cooperative computing (offloading partial computing task to the SUs) rather than performing the entire task locally, the \textit{second} term represents the total payment of the DU to all the selected cooperative SUs, and, thus the difference between the \textit{first} term (the income of the DU) and the \textit{second} term (the cost of the DU) just represents the profit that the DU can obtain under the proposed cooperative computing framework. Additionally, the utility function also takes the resource substitutability (RS) into account in the \textit{third} term as in \cite{10.2307/2555525}. RS is an ability of a resource demander to substitute one resource (hold by some oligopolies) with other resources (hold by other oligopolies) of similar functionality \cite{2019How}. Parameter $v\in[0,1]$ represents the RS. When $v=0$, there is no substitutability between resources sold by different oligopolies, when $v=1$, the resource sold by different oligopolies are completely homogeneous with each other. For example, when a DU is faced with multiple SUs to offload the computing task, if special hardware or software are needed to perform the task, the RS of different SUs is weak, namely $v$ tends to 0; otherwise, the computing task of the DU can be performed in any configured hardware or software system, namely, the computing resources of different SUs have a strong RS, so $v$ tends to 1. The RS has also been used by the authors in [7]  to address the spectrum allocation problem in multi-user cognitive radio networks.

The SUs in the game are modeled as \textit{sellers} who aim to not only earn the payment to offset the energy consumption for resource provision but also gain as much extra profits as possible. Then, the utility function for SU $n$ is defined as
\begin{equation}
    U_n=q_nl_0^n-\left(E_n^\text{rec} + \left(E_n^\text{com}-E_n^\text{exe}\right)\right),\ \forall n\in{\mathcal{N}},
\end{equation}
where the \textit{first} term represents the revenue that SU $n$ can receive from the DU, the \textit{second} term represents the extra energy consumed by SU $n$ when providing computing and communication resources to the DU rather than solely performing its own task, and thus, the difference between the \textit{first} term (the income of SU $n$) and the \textit{second} term (the cost of SU $n$) just represents the profit that SU $n$ can obtain under the proposed cooperative computing framework.

Let $\pi=\left(l_0^1,\cdots,l_0^{\left|\mathcal{N}\right|}\right)$ and $\rho=\left(q_1,\cdots,q_{\left|\mathcal{N}\right|}\right)$ denote the strategy profile of the DU and the SUs in the game, respectively. The DU and each of the SUs aim to maximize their utilities in the game by choosing the optimal strategy. Therefore, the objective of the DU is formulated as 
\begin{alignat}{3}
    &\max_{\pi \; } \;\; && U_0 \label{maxDU} \\
    &\mbox{s.t.} && 0\leq l_0^n\leq L_0, \; \forall n\in{\mathcal{N}}, & \tag{\ref{maxDU}.1} \label{positivel} \\
    &&& \sum\nolimits_{n\in \mathcal{N}}l_0^n \leq L_0, & \tag{\ref{maxDU}.2}  \label{sumconstraint}\\
    &&& p_0^n \leq P , \; \forall n\in{\mathcal{N}}, & \tag{\ref{maxDU}.3}  \label{pmax}
\end{alignat}
where constraint \eqref{sumconstraint} ensures that the amount of the offloaded data does not exceed the total amount of task data of the DU in any slot $t$, and constraint \eqref{pmax} ensures that the transmit power of the DU does not exceed the maximum allowable power $P$. The objective of SU $n$ is formulated as
\begin{alignat}{3}
    &\max_{q_n} \;\; && U_n, \; \forall n\in{{\mathcal{N}}} & \label{maxSU} \\
    &\mbox{s.t.} && q_n \geq 0, \; \forall n\in{\mathcal{N}}, & \tag{\ref{maxSU}.1}  \label{positiveq} \\
    &&& C_n(L_n+l_0^n)/T\leq f_n^\text{max}, \; \forall n\in{\mathcal{N}}, & \tag{\ref{maxSU}.2} \label{fmax}
\end{alignat}
where constraint \eqref{positiveq} ensures a positive unit-price, and constraint \eqref{fmax} is the maximum computational power constraint.

Obviously, the DU and the SUs have conflict objectives in the game. The objective is to find the NE of the game, at which no user can achieve better utility by unilaterally violating the NE strategy profile $\{\hat{\pi},\hat{ \rho}\}$.

\section{Solution of the game}
In the formulated Bertrand game, the amount of resource that the DU buys is determined by the pricing of the SUs, and vice versa. Additionally, the pricing strategy $q_n$ of SU $n$ is affected not only by its own available resource but also by the pricing of the other SUs in set $\mathcal{N}$, which is represented by $\rho_{-n}=\left(q_1,\cdots,q_{n-1},q_{n+1},\cdots,q_{\left|\mathcal{N}\right|}\right)$. Next, with the \textit{complete information} hypothesis, that is, the strategies of all the UEs are observable, we try to solve the NE strategy profile $\{\hat{\pi},\hat{\rho}\}$ of the game. The main methods are given below.
\begin{enumerate}
    \item Given the pricing strategy $\rho$ of the SUs, the DU can obtain the optimal strtegy $\hat{\pi}$ by solving problem \eqref{maxDU}.
    \item After obtaining $\hat{\pi}$, the SUs in set $\mathcal{N}$ can obtain their optimal pricing strategy $\hat{\rho}$ by solving problem \eqref{maxSU}.
\end{enumerate}

\subsection{Optimal Resource Purchase of the DU}
In this step, we solve problem \eqref{maxDU} by using given pricing $\rho$ of the SUs. Note that constraint \eqref{sumconstraint} is not considered in this step, so $\sum\nolimits_{n\in \mathcal{N}}l_0^n > L_0$ is allowed, which means that the DU can buy additional computing resources beyond the demand $L_0$ from the SUs. This constraint will be dealt with in Sec. IV. C, where the active SU set $\mathcal{N}$ is determined.

In order to obtain the analytic solution of problem \eqref{maxDU}, one can simplify eq. \eqref{U0} into
\begin{small}
\begin{align}
    U_0=-\frac{1}{2}\left(\sum\limits_{n\in \mathcal{N}}\left(\frac{H_2}{g_0^n}+1\right)\left(l_0^n\right)^2+2v\sum\nolimits_{n\neq k}l_0^nl_0^k\right)+ \notag\\ 
    \sum\nolimits_{n\in \mathcal{N}}\left(A-\frac{H_1}{g_0^n}-q_n\right)l_0^n,
\end{align}
\end{small}where $A=\kappa_0 (f_0^\text{max} )^2 C_0$, $H_1=\ln{2^{\frac{1}{BT/{\left|\mathcal{N}\right|}}}}\sigma^2T/{\left|\mathcal{N}\right|}$ and $H_2=\left(\ln{2^{\frac{1}{BT/{\left|\mathcal{N}\right|}}}}\right)^2\sigma^2T/{\left|\mathcal{N}\right|}$ are constants. The derivation is detailed in Appendix A.

Next, one can differentiate $U_0$ with respect to $l_0^n$ and solve equation $\frac{\partial U_0}{\partial l_0^n}=0$. Then, the optimal solution of problem \eqref{maxDU} is obtained as
\begin{align}
  {\hat{l_0^{n}}}=\left[\alpha_n-\beta_nq_n\right]_0^{Q_n^1},\ \forall n\in{{\mathcal{N}}},
  \label{ln}
\end{align}
\begin{small}
\begin{align}
&\text{where} \ Q_n^1=\min\left(L_0,\log_2\left(\frac{Pg_0^n}{\sigma_0^2}+1\right)\frac{BT}{{\left|\mathcal{N}\right|}}\right), \notag \\
  &\alpha_n=\frac{A-\frac{H_1}{g_0^n}\left(v\left(K-\frac{1}{{H_2}/{g_0^n}-v+1}\right)+1\right)+v\sum\nolimits_{n\neq k}\left(\frac{{H_1}/{g_0^k}+q_k}{{H_2}/{g_0^k}-v+1}\right)}{\left({H_2}/{g_0^n}-v+1\right)\left(vK+1\right)}, \notag \\
  &\beta_n=\frac{v\left(K-\frac{1}{{H_2}/{g_0^n}-v+1}\right)+1}{\left({H_2}/{g_0^n}-v+1\right)\left(vK+1\right)} \ \text{and} \ K=\sum_{n\in \mathcal{N}}\frac{1}{{H_2}/{g_0^n}-v+1} \notag
\end{align}
\end{small}are constants when all $q_k$ are known for $\forall k\in{\mathcal{N}}$ and $k\neq n$, and $\left[x\right]_a^b=\max\left(\min(x,b),\ a\right)$.

\subsection{Optimal Pricing Strategies of the SUs}

By substituting the solution of problem \eqref{maxDU} 
given in eq. \eqref{ln} into problem \eqref{maxSU}, problem \eqref{maxSU} is rewritten as

\begin{alignat}{3}
    &\max_{q_n} \;\;  U_n=q_n{\hat{l_0^{n}}}-p_n^\text{rec} t_n-\frac{\kappa_n C_n^3 \left((L_n+{\hat{l_0^{n}}})^3-L_0^3\right)}{T^2},\notag \\
    &\forall n\in{\mathcal{N}}  \label{maxSU2} \\
    &\mbox{s.t.}  \;\; q_n \geq 0, \; \forall n\in{\mathcal{N}},  \tag{\ref{maxSU2}.1}  \\
    &  \;\; \; \; C_n(L_n+{\hat{l_0^{n}}})/T\leq f_n^\text{max}, \; \forall n\in{\mathcal{N}}.  \tag{\ref{maxSU2}.2} \label{fmax2}
\end{alignat}

By combining eq. \eqref{ln} and constraint \eqref{fmax2}, we can get the more accurate range of ${\hat{l_0^{n}}}$ as 
\begin{align}
  {\hat{l_0^{n}}}=\left[\alpha_n-\beta_nq_n\right]_0^{{Q_n}},  \label{opln}
\end{align}
where $Q_n=\min\left(Q_n^1,Q_n^2\right) $, and $Q_n^2=\frac{Tf_n^\text{max}}{C_n}-L_n$.
By using eq. \eqref{opln}, the pricing range of SU $n$ is obtained as
\begin{equation}
    q_n\in\left[\frac{\alpha_n-Q_n}{\beta_n},\frac{\alpha_n}{\beta_n}\right],
    \label{pricerange}
\end{equation}
which means that SU $n$ cannot benefit by choosing prices higher than $\frac{\alpha_n}{\beta_n}$ or lower than $\frac{\alpha_n-Q_n}{\beta_n}$. Next, in order to solve problem \eqref{maxSU}, we give the following Lemma.
\newtheorem{theorem}{Lemma}
\begin{theorem}
        When $\hat{l_0^n}=\alpha_n-\beta_nq_n$ and $q_n$ satisfies condition \eqref{pricerange}, $U_n$ is concave with respect to $q_n$.
        \label{concave}
\end{theorem}

\begin{proof}
Please refer to Appendix B.
\end{proof}
According to \textbf{Lemma 1}, one can differentiate $U_n$ with respect to $q_n$ and let $\frac{\partial{U_n}}{\partial{q_n}}=0$. Then, the optimal solution to problem \eqref{maxSU2} is obtained as
\begin{equation}
\hat{q_n}=\left\{
\begin{array}{rcl}
\frac{\alpha_n-Q_n}{\beta_n}     &    & {0\leq \mu_n< \frac{\alpha_n-Q_n}{\beta_n}}\\
\mu_n  &      & {\frac{\alpha_n-Q_n}{\beta_n} \leq \mu_n \leq \frac{\alpha_n}{\beta_n}}\\
\frac{\alpha_n}{\beta_n}    &      & {\mu_n > \frac{\alpha_n}{\beta_n}}
\end{array} \right.
\end{equation}
where $\mu_n=\frac{-\sqrt{6F_nL_n\beta_n+3F_n\alpha_n\beta_n+1}+3L_nF_n\beta_n+3F_n\alpha_n\beta_n+1}{3F_n\beta_n^2}$, and $F_n=\frac{\kappa_n C_n^3 }{T^2}$. The derivation is in Appendix C.

\subsection{Algorithm Implementation}

In this section, we first consider the \textit{complete information} game (CIG), in which the channel state information (CSI) of the UEs, the pricing strategy of the SUs, and the resource purchase of the DU are all observable to the UEs. It is noted that the pricing $q_n$ of each SU $n$ is the function of $g_0^n$, i.e., the CSI between it and the DU. Therefore, the SU selection and scheduling algorithm requires to obtain the CSI via the dedicated feedback channel. According to the hierarchical structure of the game, we propose an iterative algorithm to find the NE strategy profile $\{\hat{\pi},\hat{\rho}\}$ of the CIG which is given in \textbf{Algorithm 1}. We use $i=1,2,\cdots$ to represent an index sequence of the number of iterations. In the $i$th iteration, we use $q_n[i]$, $\rho_{-n}[i]$ and $\pi[i]$ to represent the pricing of SU $n$, the pricing of the SUs in set $\mathcal{N}$ other than SU $n$, and the resource purchase of the DU, respectively. 

    \begin{algorithm}
\caption{Solving CIG}
\begin{small}
\begin{algorithmic}[1]
\STATE Let $i=1$. Initialize the pricing of the SUs in set $ \mathcal{N}$ as $\rho[i]$.
\STATE The SUs broadcast their CPU parameters $f_n^\text{max}$ to the DU.
\STATE With the given $\rho[i]$, the DU obtain the optimal resource purchase $\pi[i]$ by using eq. (16).
\REPEAT 
\STATE For each SU $n$ in set $\mathcal{N}$, after collecting $\pi[i]$ from the DU and $\rho_{-n}[i]$ from the other SUs in set $\mathcal{N}$, SU $n$ obtains its optimal pricing $q_n[i]$ by using eq. (18).
\STATE For the DU, after collecting $\rho[i]$ from the SUs in set $\mathcal{N}$, it obtains the optimal resource purchase $\pi[i]$ by using eq. (16).
\STATE Each SU can obtain the gradient of its utility $\nabla U_n[i](q_n[i])$ 
\STATE Update $i=i+1.$
\UNTIL $\left\|\nabla U_n[i](q_n[i])\right\| \leq \epsilon \left\|\nabla U_n[i-1](q_n[i-1])\right\|, \ \forall n\in{\mathcal{N}} $
\end{algorithmic}
\end{small}
\end{algorithm}

Considering the system is composed of one DU and two SUs, the convergence condition of \textbf{Algorithm 1} is analyzed in the following \textbf{Lemma 2}.

\begin{theorem}
\textbf{Algorithm 1} can converge to a stable point $\left\{\pi[i]=\hat{\pi},\rho[i]=\hat{\rho}\right\}$ as the number of iterations $i$ increases.
\end{theorem}
\begin{proof}
Please refer to Appendix D.
\end{proof}

The convergence speed of \textbf{Algorithm 1} mainly depends on the choice of the convergence threshold $\epsilon$. Since $U_n(q_n)$ is strictly concave with respect to $q_n$, according to \cite{doi:10.1080/10556788.2016.1278445}, the general upper bound on the number of iterations of \textbf{Algorithm 1} to reach a certain convergence threshold $\epsilon$ is $\mathcal{O}\left(\log(1/\epsilon)\right)$. For example, when the initial price is chosen arbitrarily from the feasible domain and the threshold $\epsilon$ is set to $10^{-3}$, \textbf{Algorithm 1} converges to a stable point after 10 iterations. It agrees to the simulation results as shown in Fig. 1. 
In addition to the CSI, the implementation of \textbf{Algorithm 1} also requires the DU and SUs to exchange $\left\{\pi[i],\rho[i]\right\}$ in any $i$th iteration. However, this information is private and difficult to obtain in practical application. Next, we consider a more realistic \textit{incomplete information} game (ICIG), in which the only information available for any SU $n$ to make decision is $l_0^n\left(i \right)$, that is the amount of resource purchased by the DU from SU $n$ in any $i$th iteration.

Next, we propose a distributed iterative algorithm based on projected gradient descent (PGD) \cite{web} to find the NE of the ICIG. Firstly, we define
\begin{small}
\begin{align}
    \frac{\partial{U_n}}{\partial{q_n}}\approx\frac{U_n\left(q_n[i]+\delta\right)-U_n\left(q_n[i]-\delta\right)}{2\delta}, \ \forall n\in\mathcal{N},
\end{align}
\end{small}where $\delta$ is a sufficiently small positive number (e.g., $\delta=10^{-5}$). According to the rules of PGD, any SU $n$ should adjust its strategy in the direction that maximizes its own utility but not beyond the feasible domain. According to constraint (12.1), we define the feasible domain of the pricing strategies of the SUs as $\mathcal{X}=\left\{x|x\geq0\right\}$. Then, the price update strategy for SU $n$ in any $i$th iteration is designed as
\begin{small}
\begin{align}
    q_n[i+1]=\mathop{\arg\min}_{x\in\mathcal{X}} \left\|x-\left(q_n[i]+a_n\left(\frac{\partial{U_n}}{\partial{q_n}}\right)\right)\right\|_2^2, \ \forall n\in\mathcal{N},
    \label{GDM}
\end{align}
\end{small}where $a_n$ is the adjustment speed (i.e., learning rate). The convergence of the PGD is analyzed in \cite{web}, thus omitted here.

The last issue to be addressed is to apply constraint \eqref{sumconstraint} to the game, that is, to determine which SUs are in set $\mathcal{N}$ and solve the problem of \textit{who to cooperate}. One simple and straightforward approach is to initialize $\mathcal{N}=\mathcal{M}-\left\{0\right\}$. After performing the game, one can make the following choices according to the obtained NE. If the NE satisfies constraint \eqref{sumconstraint} and $l_0^n>0$ for $\forall n \in \mathcal{N}$, the algorithm terminates and the NE is taken as the final solution of the game. Otherwise, the SU with the highest price is removed from set $\mathcal{N}$, and the game is performed again until constraint \eqref{sumconstraint} is satisfied or set $\mathcal{N}$ is empty. Since the the elements of the potential SU set $\mathcal{M}-\left\{0\right\}$ is limited, this algorithm is bound to terminates. The detail of the algorithm is given in \textbf{Algorithm 2}.  

\begin{algorithm}
    \renewcommand{\algorithmicrequire}{\textbf{Input:}}
    \renewcommand{\algorithmicensure}{\textbf{Output:}}
    \caption{SU selection algorithm}
    \label{alg.remove}
    \begin{small}
  \begin{algorithmic}[2]
    
    \REQUIRE The potential SUs set $\mathcal{M}-\left\{0\right\}$ for the DU.
    
    \STATE Initialization: Let $\mathcal{N}=\mathcal{M}-\left\{0\right\}$.
		\REPEAT
		\STATE Perform the game and obtain the NE $\{\hat{\pi},\hat{\rho}\}$ of the game.
		\STATE Remove any SU $n$ with $\hat{l_0^n}=0$ from set $\mathcal{N}$.
		 \IF{$\sum\nolimits_{n\in \mathcal{N}}l_0^n > L_0$}
		    \STATE  Remove the SU with the highest price from set $\mathcal{N}.$
		\ENDIF
		\UNTIL $\sum\nolimits_{n\in \mathcal{N}}l_0^n \leq L_0$ and $l_0^n>0$ for $\forall n \in \mathcal{N}$, or $\mathcal{N}=\varnothing.$
		\ENSURE The active SU set $\mathcal{N}$ for the DU.
	
  \end{algorithmic}
  \end{small}
\end{algorithm}

\section{Simulation results}
In the simulations, the system parameters are set as below. The duration of a time slot is $T=0.2$ s. The system bandwidth is $B=1$ MHz. All the UEs are with the same switch capacitance coefficient $\kappa_m =10^{-28}$. The number of CPU cycles required to execute megabit data is $C_m=8\times10^8$. To complete the computing task in slot $t$, the maximum CPU frequencies allocated by the DU and SU $n$ are $f_0^\text{max}=2.4$ GHz and $f_n^\text{max}=1.5$ GHz, respectively. The maximum transmit power of the DU is $P=0.1$ W. The power of the receiving circuit of SU $n$ is $p_n^\text{rec}=0.01$ W. The path loss gain is set to $0.001/d^3$ (where $d$ is the
distance between the transmitter and the receiver (in meters)). The noise power at the receiver of a SU is $\sigma^2 = 10^{-9}$. The substitutability factor for the DU's utility function is set to $v=0.5$.

To testify the convergence of the proposed algorithms, we place one DU at coordinate $(0,0)$, and two SUs at coordinates $(-20,20)$ and $(20,20)$, respectively. The amount of task data of the DU and the SUs are set to $L_0=0.6$ Mb, $L_1=0.15 $Mb, and $L_2=0$ Mb, respectively. 
Fig. 1 shows the convergence of the price of the SUs with the increase of iteration numbers.

\vspace{-5 mm}
\begin{figure}[H]
  \centering
  \includegraphics[scale=.5]{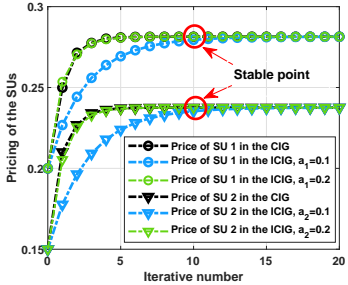} 
  \vspace{-2 mm}
  \caption{The convergence to the NE price.}
\end{figure}
\vspace{-3 mm}

From Fig. 1, one can see that \textbf{Algorithm 1} requires only a few iterations to converge to the NE in the CIG. Whereas, in the ICIG, the convergence speed depends largely on the learning rate $a_i$ in eq. \eqref{GDM}. When the learning rate is properly set, e.g., $a_1=a_2=0.2$, the ICIG converges to the NE as fast as the CIG. In addition, we note that SU 2 has a price advantage over SU 1, because SU 2 is with more idle computing resources than SU 1 in the current slot.

Next, Fig. 2 and Fig. 3 respectively show the convergence of the amount of offloaded data from the DU to the SUs and the convergence of their utilities in the ICIG.
\vspace{-4 mm}
\begin{figure}[htbp]
  \begin{minipage}[t]{0.48\linewidth}
  \centering
  \includegraphics[scale=0.43]{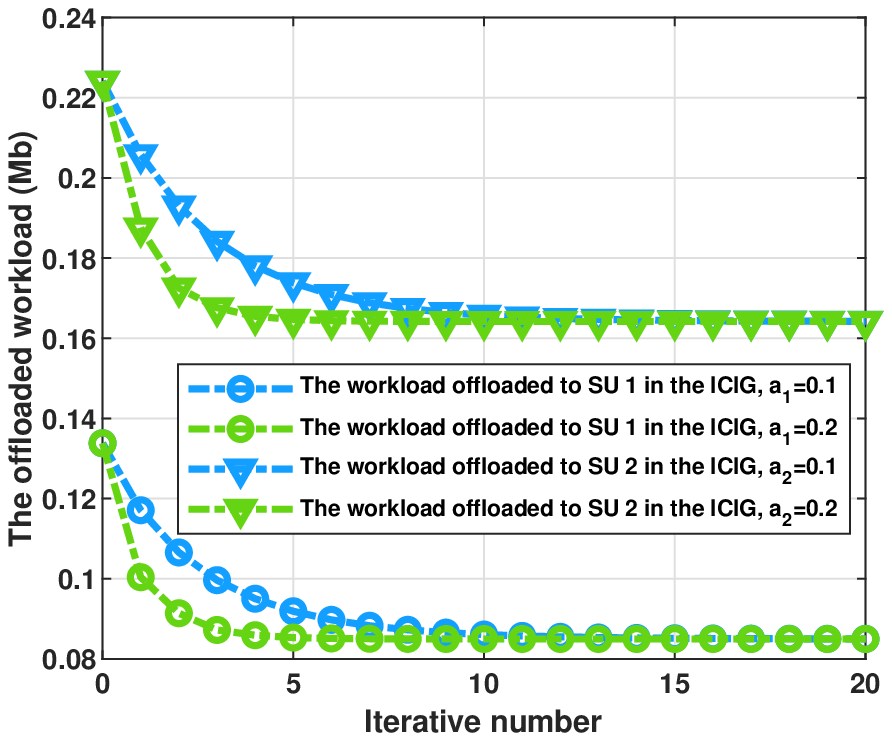}
  \vspace{-7mm}
  \caption{The convergence of the amount of offloaded workload.}
  \label{}
  \end{minipage}%
  \hspace{2mm}
  \begin{minipage}[t]{0.48\linewidth}
  \centering
  \includegraphics[scale=0.43]{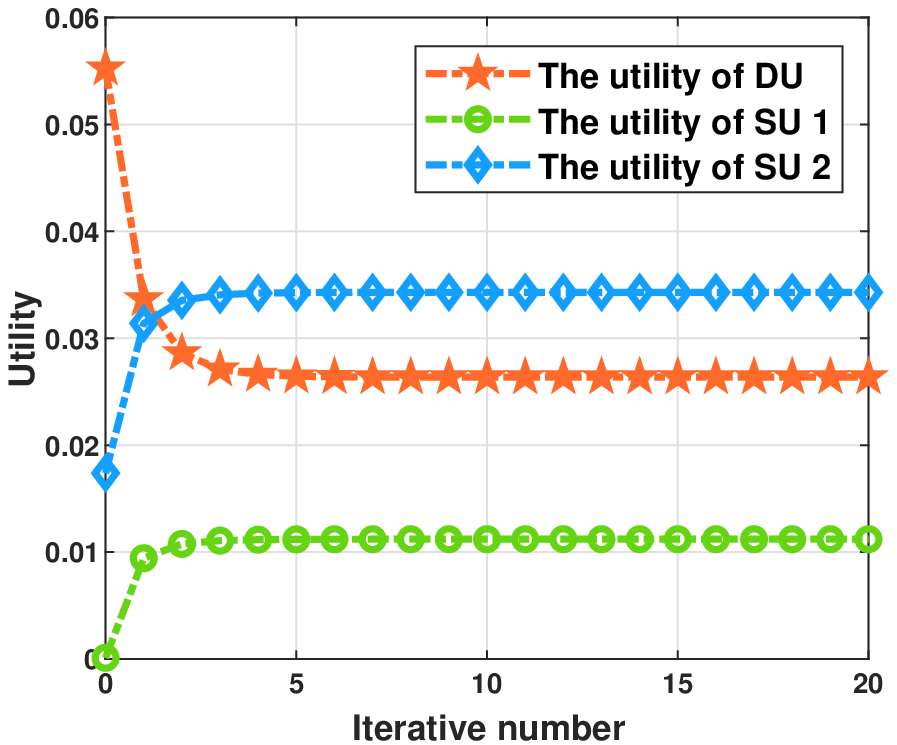}
  \vspace{-7mm}
  \caption{The convergence of user utilities.}
  \end{minipage}
  \end{figure}
\vspace{-2 mm}

From Fig. 2, one can see that SU 2 accepts more computing tasks of the DU than SU 1 in the game. This indicates that the proposed pricing game can coordinate the resource supply of the SUs according to their current available resources. From Fig. 3, one can see that all the DU and SUs obtain positive utilities from the game. Since SU 2 sells more resources than SU 1, it obtains a higher utility than SU 1. It implies that the proposed game provides sufficient motivation to rational users to participate in cooperative computation.

Finally, we simulate a system consisting of one DU and three SUs. The coordinates of the DU and SUs are $(0,0)$, $(-20,20)$, $(20,20)$ and $(20,-20)$, respectively. The amount of task data of SUs 1 and 2 are $L_1=0.15$ Mb and $L_2=0.1$ Mb, respectively. We increase the workload of SU 3 from $0$ Mb to $0.15$ Mb at a step of $0.05$ Mb. Fig. 4 shows the variation of the amount of task data offloaded from the DU to the SUs.

\vspace{-4 mm}
\begin{figure}[H]
  \centering
  \includegraphics[scale=.45]{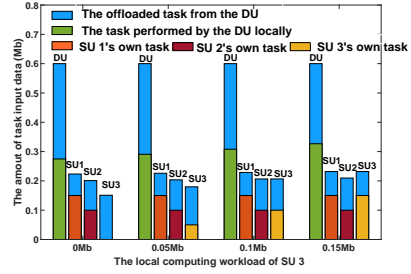} 
  \vspace{-2 mm}
  \caption{The workload distribution on the users.}
\end{figure}
\vspace{-3 mm}

From Fig. 4, one can see that with the increase of the workload of SU 3, the amount of task data offloaded from the DU to SU 3 decreases gradually. At the same time, the workloads offloaded from the DU to SU 1 and SU 2 have a small increase. It indicates that the proposed game can coordinate the resource provision of the SUs according to their current available resources.  

\section{Conclusion}
In this paper, an incentive scheme based on Bertrand game has been proposed to stimulate rational users to participate in cooperative computation framework. When the game is with \textit{complete information}, the NE of the game is obtained in closed-form, while in the case of \textit{incomplete information}, a distributed iterative algorithm has been developed to find the NE. Simulation results have verified the effectiveness of the proposed scheme.

\appendices
\section{Derivation of Function (13)}
A Maclaurin series is a function which has expansion series that gives the sum of derivatives of that function. In order to obtain the analytic solution of problem (11), we derive the Maclaurin series of the second term of function (9) up to order $n=2$. Then function (9) is rewritten as
\begin{small}
\begin{align}
    &U_0=\kappa_0 (f_0^\text{max} )^2 C_0 \sum\limits_{n\in \mathcal{N}}l_0^n -\frac{\sigma^2T/{\left|\mathcal{N}\right|}}{g_0^n}\sum\limits_{n\in \mathcal{N}} \left(l_0^n\ln{2^{\frac{1}{BT/{\left|\mathcal{N}\right|}}}} \right. \notag\\
   &\left. +\frac{1}{2}\left(l_0^n\ln{2^{\frac{1}{BT/{\left|\mathcal{N}\right|}}}}\right)^2\right) -\frac{1}{2}\left(\sum\limits_{n\in \mathcal{N}}\left({l_0^n}\right)^2+2v\sum_{n\neq k}l_0^nl_0^k\right) \notag\\
    &-\sum\limits_{n\in \mathcal{N}}q_nl_0^n. \tag{S.1}
    \label{Uu0}
\end{align}
\end{small}

After some algebraic manipulation, we can further rewrite \eqref{Uu0} as
\begin{align}
    U_0=-\frac{1}{2}\left(\sum\limits_{n\in \mathcal{N}}\left(\frac{H_2}{g_0^n}+1\right)\left(l_0^n\right)^2+2v\sum\nolimits_{n\neq k}l_0^nl_0^k\right) \notag \\
    +\sum\nolimits_{n\in \mathcal{N}}\left(A-\frac{H_1}{g_0^n}-q_n\right)l_0^n, \tag{S.2}
\end{align}
where $A=\kappa_0 (f_0^\text{max} )^2C_0$, $H_1=\ln{2^{\frac{1}{BT/{\left|\mathcal{N}\right|}}}}\sigma^2T/{\left|\mathcal{N}\right|}$, $H_2=(\ln{2^{\frac{1}{BT/{\left|\mathcal{N}\right|}}}})^2\sigma^2T/{\left|\mathcal{N}\right|}$ are constants. 


\section{Proof of Lemma 1}
We take the first order derivative of $U_n$ in eq. (15) with respect to $q_n$, and we have
\begin{align}
    &\frac{\partial{U_n}}{\partial{q_n}}={\hat{l_0^{n}}}+q_n\frac{\partial{{\hat{l_0^{n}}}}}{\partial{q_n}}-3F_n\left(L_n+{\hat{l_0^{n}}}\right)^2\frac{\partial{{\hat{l_0^{n}}}}}{\partial{q_n}}, \tag{S.3}
\end{align}
where $F_n=\frac{\kappa_nC_n^3}{T^2}$.

Since ${\hat{l_0^{n}}}=\alpha_n-\beta_nq_n$ is considered, we have
\begin{align}
    &\frac{\partial{U_n}}{\partial{q_n}}={\hat{l_0^{n}}}-q_n\beta_n+3F_n\beta_n\left(L_n+{\hat{l_0^{n}}}\right)^2. \tag{S.4}
\end{align}

Next, we further derive the second order derivative of eq. (15) with respect to $q_n$ as
\begin{align}
    &\frac{\partial^2{U_n}}{\partial{q_n^2}}=-2\beta_n-6F_n\beta_n^2\left(L_n+{\hat{l_0^{n}}}\right). \tag{S.5}
\end{align}

Since $F_n>0$, $\frac{\partial{{\hat{l_0^{n}}}}}{\partial{q_n}}=-\beta_n< 0$, $\frac{\partial^2{{\hat{l_0^{n}}}}}{\partial{q_n^2}}=0$, and ${\hat{l_0^{n}}}\ge 0$, we can obtain $\frac{\partial^2{U_n}}{\partial{q_n^2}} < 0$. Therefore, we conclude that the utility function $U_n$ of SU $n$ is concave with respect to $q_n$.

\section{Derivation of Equation (18)}
Solving function $\frac{\partial{U_n}}{\partial{q_n}}=0$ is equivalent to solving the following quadratic equation
\begin{align}
    &\alpha_n-2\beta_nq_n+3F_n\beta_n(L_n+\alpha_n-\beta_nq_n)^2=0. \tag{S.6}
    \label{partial}
\end{align}

By solving eq. \eqref{partial}, one can get the following two solutions, i.e., 
\begin{small}
\begin{align}
    &q_n=\mu_{n,1}=\frac{3L_nF_n\beta_n+3F_n\alpha_n\beta_n+1+\sqrt{6F_nL_n\beta_n+3F_n\alpha_n\beta_n+1}}{3F_n\beta_n^2} \tag{S.7}
\end{align}
\end{small}
and
\begin{small}
\begin{align}
    &q_n=\mu_{n,2}=\frac{3L_nF_n\beta_n+3F_n\alpha_n\beta_n+1-\sqrt{6F_nL_n\beta_n+3F_n\alpha_n\beta_n+1}}{3F_n\beta_n^2} \tag{S.8} 
\end{align}
\end{small}

From the first solution, we can derive $\mu_{n,1}>\frac{\alpha_n}{\beta_n}$. It means that this solution is beyond the range of feasible prices, resulting in the amount of input data of the offloaded task being 0 and, hence, $U_n=0$. So equation \eqref{partial} has the unique solution $q_n=\mu_{n,2}$. For easy expression, we define $\mu_n=\mu_{n,2}$. Then we have $q_n=\mu_n$.

Now, we determine the solution of problem (15) according to the feasible price range where $\mu_n$ is located.
\begin{itemize}
    \item If $\mu_n<\frac{\alpha_n-Q_n}{\beta_n}$, $U_n$ is monotonically decreasing with respect to the $q_n$ over the feasible price range. So the optimal price is $\hat{q_n}=\frac{\alpha_n-Q_n}{\beta_n}$.
    \item If $\frac{\alpha_n-Q_n}{\beta_n} \leq \mu_n \leq \frac{\alpha_n}{\beta_n} $ , the optimal price is $\hat{q_n}=\mu_n$.
    \item If $\mu_n > \frac{\alpha_n}{\beta_n}$, $U_n$ is monotonically increasing respect to $q_n$ over the feasible price range. So the optimal price is $\hat{q_n}=\frac{\alpha_n}{\beta_n}.$
\end{itemize}

Finally, the solution of problem (15) is given as 

\begin{align}
\hat{q_n}=\left\{
\begin{array}{rcl}
\frac{\alpha_n-Q_n}{\beta_n}     &    & {0\leq \mu_n< \frac{\alpha_n-Q_n}{\beta_n}}\\
\mu_n  &      & {\frac{\alpha_n-Q_n}{\beta_n} \leq \mu_n \leq \frac{\alpha_n}{\beta_n}}\\
\frac{\alpha_n}{\beta_n}    &      & {\mu_n > \frac{\alpha_n}{\beta_n}}
\end{array} \right. \tag{S.9}
\end{align}

\section{Proof of Lemma 2}
According to eq. (18), we can obtain the self-mapping function of SU $n$ as
\begin{equation}
q_n\left[i+1\right]=\left\{
\begin{array}{rcl}
\frac{\alpha_n\left[i\right]-Q_n}{\beta_n}     &    & {0\leq \mu_n\left[i\right]< \frac{\alpha_n\left[i\right]-Q_n}{\beta_n}}\\
\mu_n\left[i\right]  &      & {\frac{\alpha_n\left[i\right]-Q_n}{\beta_n} \leq \mu_n\left[i\right] \leq \frac{\alpha_n\left[i\right]}{\beta_n}}\\
\frac{\alpha_n\left[i\right]}{\beta_n}    &      & {\mu_n\left[i\right] > \frac{\alpha_n\left[i\right]}{\beta_n}}
\end{array} \right. \tag{S.10}
\label{mapf}
\end{equation}
The Jacobian matrix of eq. \eqref{mapf} is given by

\begin{equation}       
\mathbf J=\left[                 
  \begin{array}{cc}   
    J_{1,1} & J_{1,2} \\  
    J_{2,1} & J_{2,2} \\  
  \end{array}
\right] =\left[                 
  \begin{array}{cc}   
    \frac{\partial{q_1[i+1]}}{\partial{q_1[i]}} & \frac{\partial{q_1[i+1]}}{\partial{q_2[i]}} \\  
    \frac{\partial{q_2[i+1]}}{\partial{q_1[i]}} & \frac{\partial{q_2[i+1]}}{\partial{q_2[i]}} \\  
  \end{array}
\right].                        \tag{S.11}
\end{equation}

By definition, the self-mapping function \eqref{mapf} is stable if and only if the eigenvalues $\lambda_n$ of $\mathbf J$ are all inside the unit circle of the complex plane, i.e., $\left|\lambda_n\right| < 1$.

Next, we derive the expression of the elements of $\mathbf J$ as
\begin{equation}
    J_{1,1}=J_{2,2}=\frac{\partial{q_n[i+1]}}{\partial{q_n[i]}}=0, \tag{S.12}
\end{equation}
\begin{small}
\begin{align}
J_{1,2}&=\frac{\partial{q_1[i+1]}}{\partial{q_2[i]}} \notag \\
&=\left\{
\begin{array}{rcl}
\frac{1}{\beta_1} \frac{\partial \alpha_1}{\partial q_2}    &    & {0\leq \mu_1\left[i\right]< \frac{\alpha_1\left[i\right]-Q_1}{\beta_1}}\\
\left(1-\frac{1}{2\sqrt{\zeta_1}}\right)\frac{1}{\beta_1} \frac{\partial \alpha_1}{\partial q_2}  &      & {\frac{\alpha_1\left[i\right]-Q_1}{\beta_1} \leq \mu_1\left[i\right] \leq \frac{\alpha_1\left[i\right]}{\beta_1}}\\
\frac{1}{\beta_1} \frac{\partial \alpha_1}{\partial q_2}   &      & {\mu_1\left[i\right] \geq \frac{\alpha_1\left[i\right]}{\beta_1}}
\end{array} \right.\tag{S.13}
\end{align}
\end{small}
and
\begin{small}
\begin{align}
J_{2,1}&=\frac{\partial{q_2[i+1]}}{\partial{q_1[i]}} \notag \\
&=\left\{
\begin{array}{rcl}
\frac{1}{\beta_2} \frac{\partial \alpha_2}{\partial q_1}    &    & {0\leq \mu_2\left[i\right]< \frac{\alpha_2\left[i\right]-Q_2}{\beta_2}}\\
\left(1-\frac{1}{2\sqrt{\zeta_2}}\right)\frac{1}{\beta_2} \frac{\partial \alpha_2}{\partial q_1}  &      & {\frac{\alpha_2\left[i\right]-Q_2}{\beta_2} \leq \mu_2\left[i\right] \leq \frac{\alpha_2\left[i\right]}{\beta_2}}\\
\frac{1}{\beta_2} \frac{\partial \alpha_2}{\partial q_1}   &      & {\mu_2\left[i\right] \geq \frac{\alpha_2\left[i\right]}{\beta_2}}
\end{array} \right.\tag{S.14}
\end{align}
\end{small}
In eq. (S.13) and eq. (S.14), we have
\begin{align}
\zeta_1=6L_1F_1\beta_1+3F_1\alpha_1\beta_1+1, \tag{S.15} \\
\zeta_2=6L_2F_2\beta_2+3F_2\alpha_2\beta_2+1, \tag{S.16} \\
\frac{1}{\beta_1} \frac{\partial \alpha_1}{\partial q_2}=\frac{\frac{v}{H_2/g_0^2-v+1}}{\frac{v}{H_2/g_0^2-v+1}+1}<1  \tag{S.17}
\end{align}
and
\begin{align}
\frac{1}{\beta_2} \frac{\partial \alpha_2}{\partial q_1}=\frac{\frac{v}{H_2/g_0^1-v+1}}{\frac{v}{H_2/g_0^1-v+1}+1}<1. \tag{S.18}
\end{align}
Then, we know $J_{1,2}<1$ and $J_{2,1}<1$.

Since all the elements of $\mathbf J$ are less than one, the eigenvalues of $\mathbf J$ are all less than 1, which is given by
\begin{align}
    (\lambda_1,\lambda_2)=\frac{(J_{1,1}+J_{2,2})\pm \sqrt{4J_{1,2}J_{2,1}+(J_{1,1}-J_{2,2})^2}}{2}, \tag{S.19}
\end{align}
which indicates that the eigenvalues $\lambda_n$ are all inside the unit circle of the complex plane. Therefore, \textbf{Algorithm 1} can converge to a stable point.

\ifCLASSOPTIONcaptionsoff
  \newpage
\fi

\bibliographystyle{IEEEtran}
\bibliography{IEEEabrv}

\end{document}